\documentclass[11pt,letterpaper]{article}
\usepackage{color,ifpdf,latexsym,graphicx,url,wrapfig}
\usepackage[margin=1in]{geometry}
\usepackage{amsthm}
\usepackage{subcaption}
\def\defn#1{\textbf{\textit{\boldmath #1}}}
\let\emph=\defn

\title{Complexity of Retrograde and Helpmate Chess Problems: \\
       Even Cooperative Chess is Hard}
\author{%
  Josh Brunner%
    \thanks{MIT Computer Science and Artificial Intelligence Laboratory,
      32 Vassar St., Cambridge, MA 02139, USA,
      \protect\url{{brunnerj,edemaine,dylanhen,wellman}@mit.edu}}
\and
  Erik D. Demaine\footnotemark[1]
\and
  Dylan Hendrickson\footnotemark[1]
\and
  Julian Wellman\footnotemark[1]
}
\date{}

\newif\ifabstract
\abstracttrue
\newif\iffull
\ifabstract \fullfalse \else \fulltrue \fi

\usepackage{hyperref}
\hypersetup{breaklinks,bookmarksnumbered,bookmarksopen,bookmarksopenlevel=2}
{\makeatletter \hypersetup{pdftitle={\@title}}}

{\makeatletter
 \gdef\xxxmark{%
   \expandafter\ifx\csname @mpargs\endcsname\relax 
     \expandafter\ifx\csname @captype\endcsname\relax 
       \marginpar{xxx}
     \else
       xxx 
     \fi
   \else
     xxx 
   \fi}
 \gdef\xxx{\@ifnextchar[\xxx@lab\xxx@nolab}
 \long\gdef\xxx@lab[#1]#2{\textbf{[\xxxmark #2 ---{\sc #1}]}}
 \long\gdef\xxx@nolab#1{\textbf{[\xxxmark #1]}}
}

{\makeatletter \gdef\fps@figure{!htbp}}

\let\realbfseries=\bfseries
\def\bfseries{\realbfseries\boldmath}

\newtheorem{theorem}{Theorem}[section]
\newtheorem{lemma}[theorem]{Lemma}
\newtheorem{corollary}[theorem]{Corollary}
\theoremstyle{definition}

\newtheorem{problem}{Problem}

\graphicspath{{figures/}}

\begin{document}
\maketitle

\begin{abstract}
  We prove PSPACE-completeness of two classic types of Chess problems when
  generalized to $n \times n$ boards.
  A ``retrograde'' problem asks whether it is
  possible for a position to be reached from a natural starting position,
  i.e., whether the position is ``valid'' or ``legal'' or ``reachable''.
  Most real-world retrograde Chess problems ask for the last few moves of
  such a sequence; we analyze the decision question which gets at the
  existence of an exponentially long move sequence.
  A ``helpmate'' problem asks whether
  it is possible for a player to become checkmated
  by any sequence of moves from a given position.
  A helpmate problem is essentially a cooperative form of Chess, where both
  players work together to cause a particular player to win; it also arises
  in regular Chess games, where a player who runs out of time (flags) loses
  only if they could ever possibly be checkmated from the current position
  (i.e., the helpmate problem has a solution).
  Our PSPACE-hardness reductions are from a variant of a puzzle game
  called Subway Shuffle.
\end{abstract}

\section{Introduction}

\begin{wrapfigure}{r}{2.2in}
  \centering
  \includegraphics[width=\linewidth]{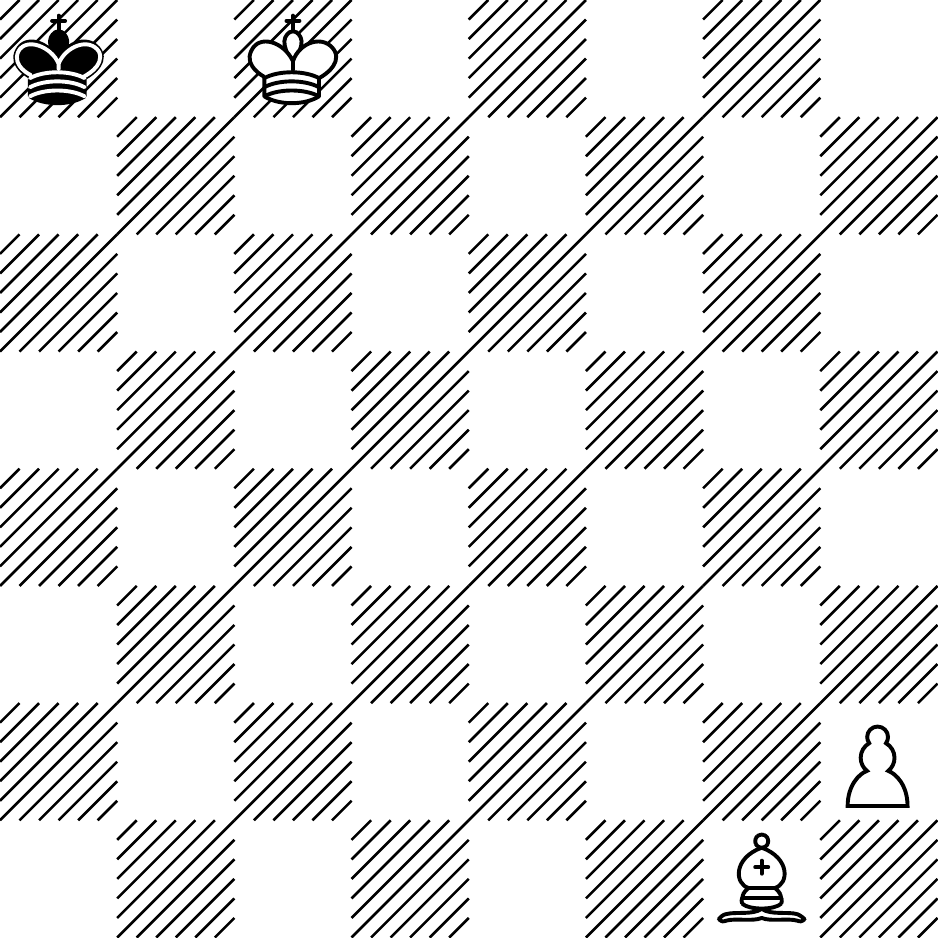}
  \caption{The retrograde Chess problem on the cover of
    \cite{Smullyan-Sherlock}.
    What move did Black just make?
    What move did White make before that?}
  \label{fig:sherlock}
\end{wrapfigure}

\defn{Chess problems}
\cite{Nunn,Smullyan-Sherlock,Smullyan-Arabian,problems-wiki}
are puzzles involving Chess boards/pieces/positions,
often used as exercises to learn how to play Chess better.
Perhaps the most common family of Chess problems are of the form
\emph{mate-in-$k$}: is it possible to force a win within $k$ moves
from the given game position (board state and who moves next)?
While this problem can be solved in polynomial time for $k = O(1)$,
it is PSPACE-complete if $k$ is polynomial in the board size $n$
\cite{Storer-1983}
and EXPTIME-complete if $k$ is exponential in the board size (or infinite)
\cite{Fraenkel-Lichtenstein-1981}.

In this paper, we analyze the complexity of two popular families of Chess
problems that are fundamentally \defn{cooperative}: they ask whether gameplay
could possibly produce a given result, which is equivalent to the two players
(Black and White) cooperating to achieve the goal.  This cooperative means the
two players effectively act as a single player (in the sense that quantifiers
no longer alternate),
placing the problem in PSPACE (see Lemma~\ref{lem:PSPACE}).
We prove that the following two problems are in fact PSPACE-complete.

First, \defn{retrograde Chess problems} ask about the moves leading
up to a given position.
For example, Figure~\ref{fig:sherlock} gives the puzzle on the cover of
Raymond Smullyan's classic book \textit{The Chess Mysteries of Sherlock Holmes}
\cite{Smullyan-Sherlock}.
Other classic books with Chess problems (and
descriptions of how to do retrograde analysis) are by Nunn \cite{Nunn} and
Smullyan \cite{Smullyan-Arabian}.
Many retrograde Chess problems (including Figure~\ref{fig:sherlock}) ask what
the final few moves of the two players must have been to reach this position.
Fundamentally, these problems are about the \defn{reachability} of the given
position from the starting position, and focus on the final $k$
moves where the puzzle is most interesting (as the moves are most forced).
For exponentially large $k$, we find a core underlying decision problem:
can the given position be reached at all from the starting position?
Such positions are often called ``valid'' or ``legal'' because they are
possible results of valid/legal gameplay; in Section~\ref{sec:reachability},
we prove that characterizing such positions is PSPACE-complete.

Second, \defn{helpmate Chess problems} \cite{chess-compendium,helpmate-wiki}
ask whether it is possible for a player to win via checkmate by any sequence
of moves, i.e., when the players cooperate (help each other, hence
``helpmate'').
This problem could also naturally be called \emph{Cooperative Chess}, by
analogy to Cooperative Checkers, which is NP-complete \cite{Checkers_MOVES2017}.
In addition to being a popular form of Chess problem,
helpmate problems naturally arise in regular games of Chess, as
FIDE's%
\footnote{The International Chess Federation (FIDE) is the governing
  body of international Chess competition.  In particular, they organize
  the World Chess Championship which defines the world's best Chess
  player.  All top-level Chess competitions (not just FIDE's)
  follow FIDE's rules of Chess \cite{FIDE}.}
``dead-reckoning'' rule says that any position without a helpmate is automatically a draw:
\begin{quote}
``The game is drawn when a position has arisen in which neither player can
checkmate the opponent's king with any series of legal moves.
The game is said to end in a `dead position'.'' \cite[Article 5.2.2]{FIDE}
\end{quote}
In practice, this condition is often checked when a player runs out of time, in which case that player loses if and only if they could ever
possibly be checkmated from the current position
(i.e., the helpmate problem has a solution) \cite[Article 6.9]{FIDE}. 
In Section~\ref{sec:helpmate},
we prove that characterizing such non-dead positions is PSPACE-complete.
Amusingly, this result implies that it is PSPACE-complete to decide whether a given
game position is already a draw (\defn{draw-in-$0$}) for Chess.

\subsection{Chess Problem Definitions}

To formalize the results summarized above, we more carefully define the
objects problems discussed in this paper.

A \emph{Chess position} is a description of an $n\times n$ square grid, where
some squares have a Chess piece (a pawn, rook, knight, bishop, queen, or king
designated either black or white) and a designation of which player (black or
white) plays next.

We follow the standard FIDE rules of Chess \cite{FIDE}, naturally generalized
to larger boards.  In particular, there must be exactly one king of each color;
colors alternate turns; a king cannot be in check after its color's turn; and
rooks, bishops, and queens can move any distance (as also generalized in
\cite{Storer-1983,Fraenkel-Lichtenstein-1981}).

To define reachability for $n \times n$ boards, we define a natural
\emph{starting position} to be a Chess position in which all of the
following conditions hold:
\begin{enumerate}
\item The first two ranks (rows) are filled with white pieces; the last two ranks are filled with black pieces; and the rest of the board is empty.
\item The second and second-to-last ranks contain only pawns.
\item The first and last ranks contain no pawns and exactly one king each,
  and sufficiently many of each of the non-pawn piece types.
  (The exact composition and ordering of these ranks will not affect our
   reduction.)
\end{enumerate}
%

%

Now we can define the two decision problems studied in this paper:

\begin{problem}[\defn{Reachability}]
  Given an $n \times n$ Chess position,
  is it possible to reach that position from a starting position?
\end{problem}

\begin{problem}[\defn{Helpmate}]
  Given an $n \times n$ Chess position,
  is it possible to reach a position in which the black king is checkmated?
\end{problem}

\begin{lemma} \label{lem:PSPACE}
  Both helpmate and reachability are in PSPACE.
\end{lemma}

\begin{proof}
  An $n\times n$ Chess position takes only polynomial (in $n$) space to record. A nondeterministic polynomial-space machine can guess a sequence of moves, accepting when it achieves checkmate (for helpmate) or reaches the target position (for reachability); thus both problems are in NPSPACE${}={}$PSPACE.
\end{proof}

We prove that both problems are in fact PSPACE-complete.
Evidence for these problems not being in NP were first given by Shitov's
examples of two legal positions that require exponentially many moves to
go between \cite{shitov2014chess}, using long chains of bishops
locked by pawns.
Our constructions to show PSPACE-hardness take on a similar flavor.

\subsection{Subway Shuffle}

Our reductions are from a one-player puzzle game called \emph{Subway Shuffle},
introduced by Hearn \cite{HD,Hearn} in his 2006 thesis, and shown
PSPACE-complete in 2015 \cite{subwayshuffle}.
Recently, Brunner et al.~\cite{rushhour} introduced a variation called
\defn{oriented Subway Shuffle} and proved it PSPACE-complete, even with
only two colors, a limited vertex set, and a single unoccupied vertex. 

Our reductions to show PSPACE-hardness of Chess-related problems are from a slightly modified version of this restricted form of oriented Subway Shuffle, which we will call ``Subway Shuffle'' for simplicity, defined as follows:

\begin{problem}[\defn{Subway Shuffle}] \label{prob:subway}
  We are given a planar directed graph with edges colored \defn{orange} and \defn{purple}, where each vertex has degree at most three and is incident to at most two edges of each color. Each vertex except one has a \defn{token}, which is also colored orange or purple. One edge is marked as the \defn{target edge}.

  A legal move is to move a token across an edge of the same color, in the direction of the edge, to an empty vertex, and then reverse the direction of the edge.

  The Subway Shuffle decision problem asks whether there is any sequence of legal moves which moves a token across the target edge.
\end{problem}

This definition differs from that in \cite{rushhour} only in the goal condition; in \cite{rushhour}, the goal is to move a specified token to a specified vertex. However, their proof of PSPACE-hardness also works for our definition, where the goal is to move a token across a specified edge; by examining the win gadget in \cite{rushhour}, it is clear that the target token can reach the target vertex exactly when a specific edge is used, so we can set that edge as the target edge.%
\footnote{This is the middle purple edge in the bottom row in Figure 6(a) in \cite{rushhour}. We also remove the target vertex (and the edge incident to it) so there is only one unoccupied vertex.}
Thus we have the following result:

\begin{theorem}[\cite{rushhour}]
  Subway Shuffle is PSPACE-complete.
\end{theorem}

\section{Helpmate Chess Problems are PSPACE-Complete}
\label{sec:helpmate}

In this section, we prove that Helpmate is PSPACE-complete by reducing from Subway Shuffle.

\begin{theorem}
  Helpmate is PSPACE-complete.
\end{theorem}

The structure of the reduction is to use a line of pieces of one type to represent and edge in the Subway Shuffle graph, where using the edge involves moving every piece in the line one space. A vertex is represented by a square where pieces from three different edges can move to. The two colors of Subway Shuffle are represented by which piece type is present in the vertex. All of the moving pieces involved in the reduction are white; black will be given a gadget to pass their turn with. In many of the figures, we label the relevant pieces that can move in red; all of the red pieces are white in the actual Chess position.

In order to make sure that players cannot make moves outside of the reduction, all of the edge and vertex gadgets are walled in with walls of bishops and pawns that are completely stuck.

\subsection{Two-Orange One-Purple Subway Shuffle}

First, we show how to modify Subway Shuffle slightly to make our reduction simpler.
In Subway Shuffle, some vertices have two orange edges incident while others have two purple edges incident. Rather than trying to build separate gadgets for each of this cases, we use Lemma~\ref{color changing lemma} to have every vertex have two orange edges and one purple edge. This way we only need to build gadgets for one type of vertex. 

\begin{lemma}
  \label{color changing lemma}
  Subway Shuffle is PSPACE-complete even when every degree three vertex has exactly two orange and one purple edge incident.
\end{lemma}
\begin{proof}
  Given an instance of Subway Shuffle, every vertex with two purple incident edges can be replaced with one with two orange incident edges as shown in Figure~\ref{fig:colorchanging}. Note that while we need to transform vertices with one or two purple outgoing edges, we don't need to worry about vertices with zero purple outgoing edges. This is because a purple vertex with zero purple outgoing edges can never move, so the entire vertex is stuck and can safely be ignored. It is easy to check that the set of legal moves is the almost the same in every configuration. The only difference is that in Figure~\ref{fig:colorchanging}(d), the purple token can leave the vertex twice through each of the two outgoing edges; however the second purple token that leaves doesn't allow any further moves except moving the purple token back into place. With the assumption that only one vertex is ever empty, this situation is never useful, so this replacement perfectly simulates the original vertex.
  \end{proof}

\begin{figure}
  \centering
  \begin{subfigure}{.3\linewidth}
    \centering
    \includegraphics[width=\linewidth]{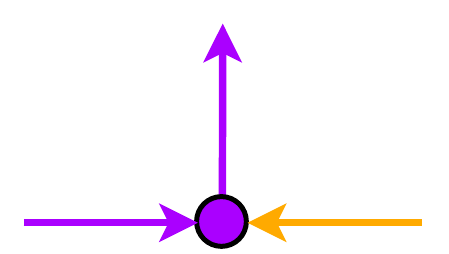}
    \caption{}
  \end{subfigure}%
  \hfil
  \begin{subfigure}{.3\linewidth}
    \centering
    \includegraphics[width=\linewidth]{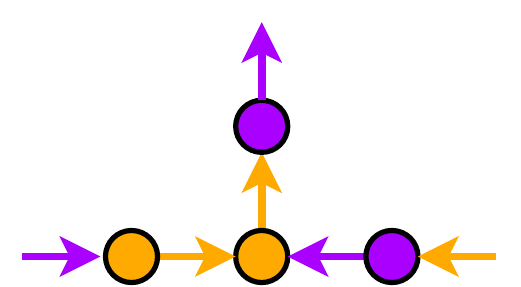}
    \caption{}
  \end{subfigure}

  \begin{subfigure}{.3\linewidth}
    \centering
    \includegraphics[width=\linewidth]{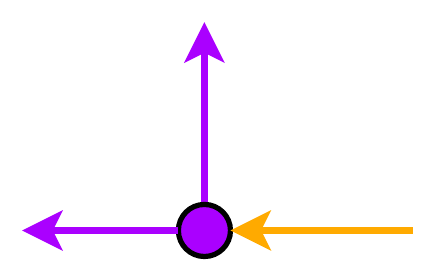}
    \caption{}
  \end{subfigure}%
  \hfil
  \begin{subfigure}{.3\linewidth}
    \centering
    \includegraphics[width=\linewidth]{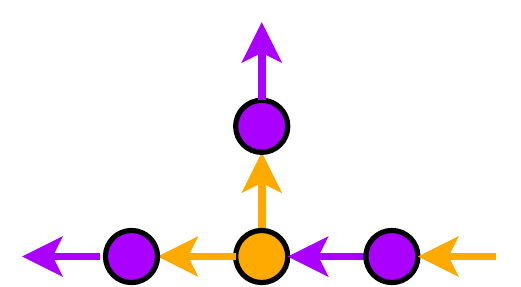}
    \caption{}
  \end{subfigure}
    \caption{The color changing gadget. (a) and (b) show the transformation for vertices with one coming purple edge, and (c) and (d) show the transformation for vertices with no incoming purple edges. In both cases, the vertex behaves identically after the change, and using this technique we can give every degree-3 vertex two orange edges.}

\label{fig:colorchanging}
\end{figure}

\begin{figure}
  \centering
  \includegraphics[width=.8\linewidth]{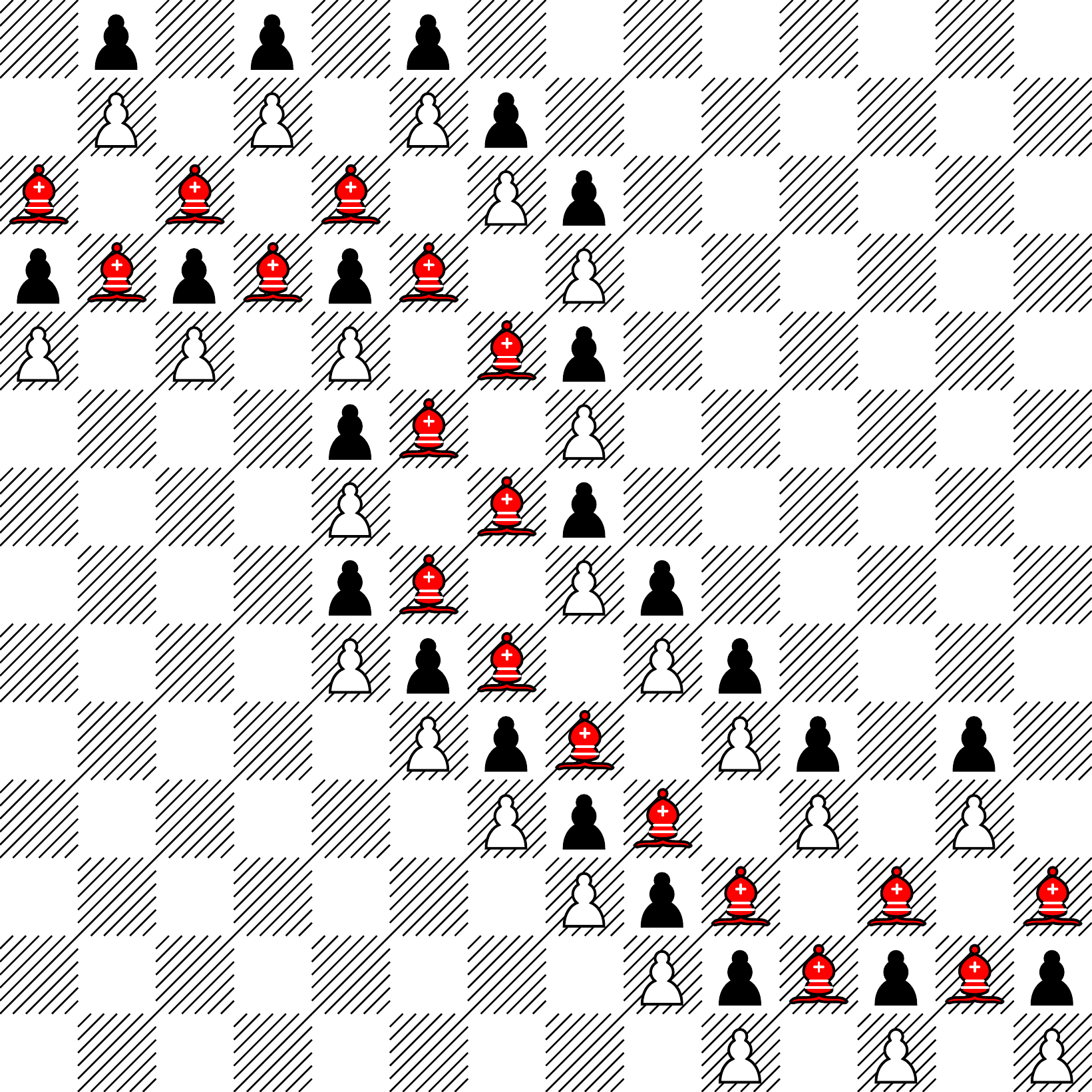}
  \caption{Our edge gadget. Since pawns care about their orientation, the gadget looks slightly different when the edge runs vertically compared to horizontally. This figure shows what both look like and how it can turn. Red represents movable white pieces.}
  \label{checkmate edge}
\end{figure}

\subsection{Gadgets}

We start with the \defn{edge gadget}. To represent an edge, we simply use a line of bishops, shown in Figure~\ref{checkmate edge}. To move a token along the edge, move all of the bishops one space in that direction. The net effect will be a bishop entering one end and another bishop leaving the other end, representing a token moving.

\begin{figure}
  \centering
  \begin{subfigure}{.47\linewidth}
    \centering
    \includegraphics[width=\linewidth]{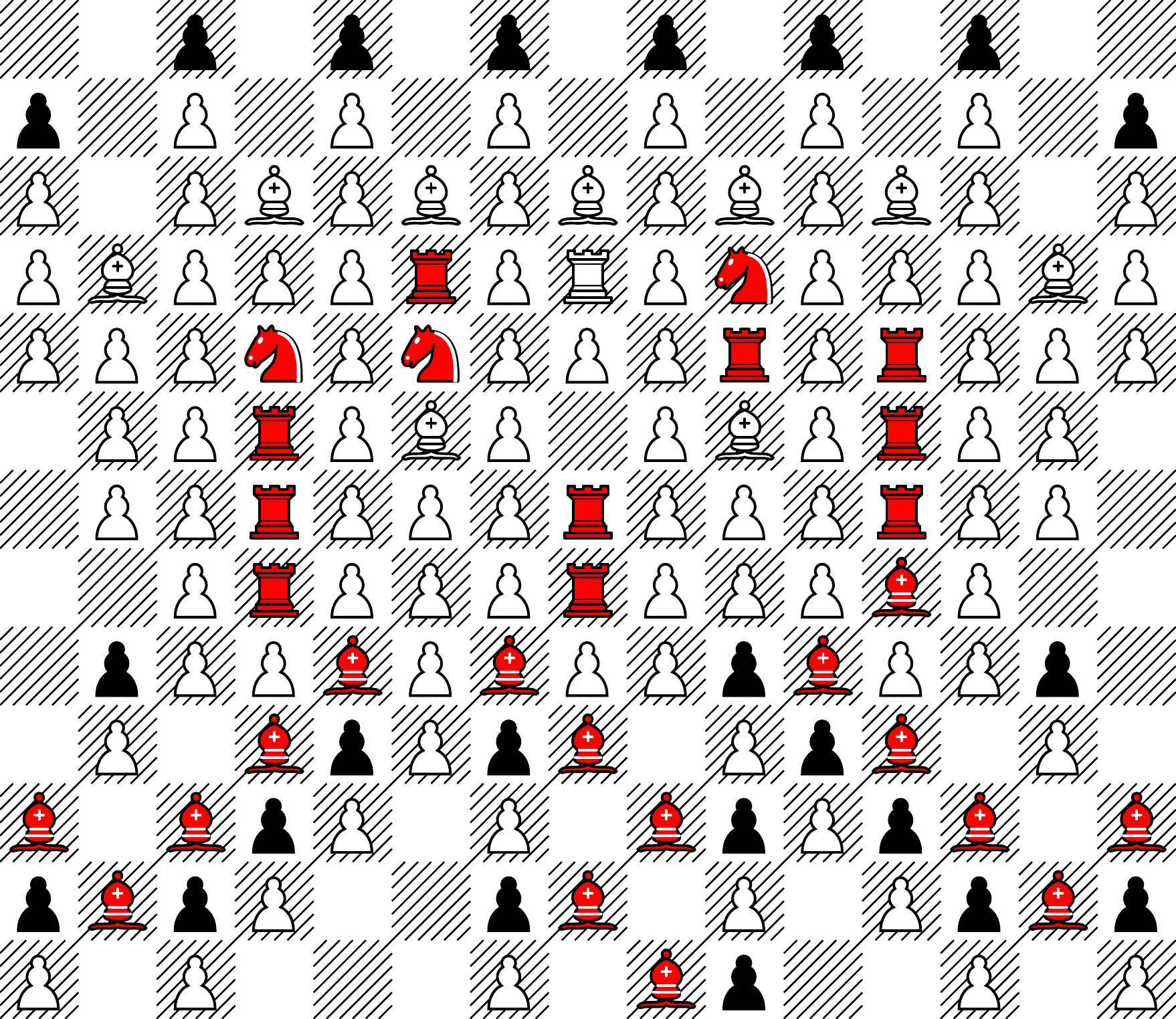}

    \medskip

    \includegraphics[scale=0.75]{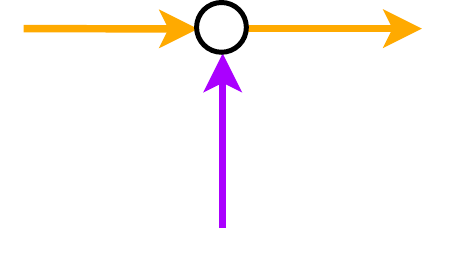}
    \caption{An empty vertex with one orange edge pointing in and one pointing out.}
  \end{subfigure}%
  \hfill
  \begin{subfigure}{.47\linewidth}
    \centering
    \includegraphics[width=\linewidth]{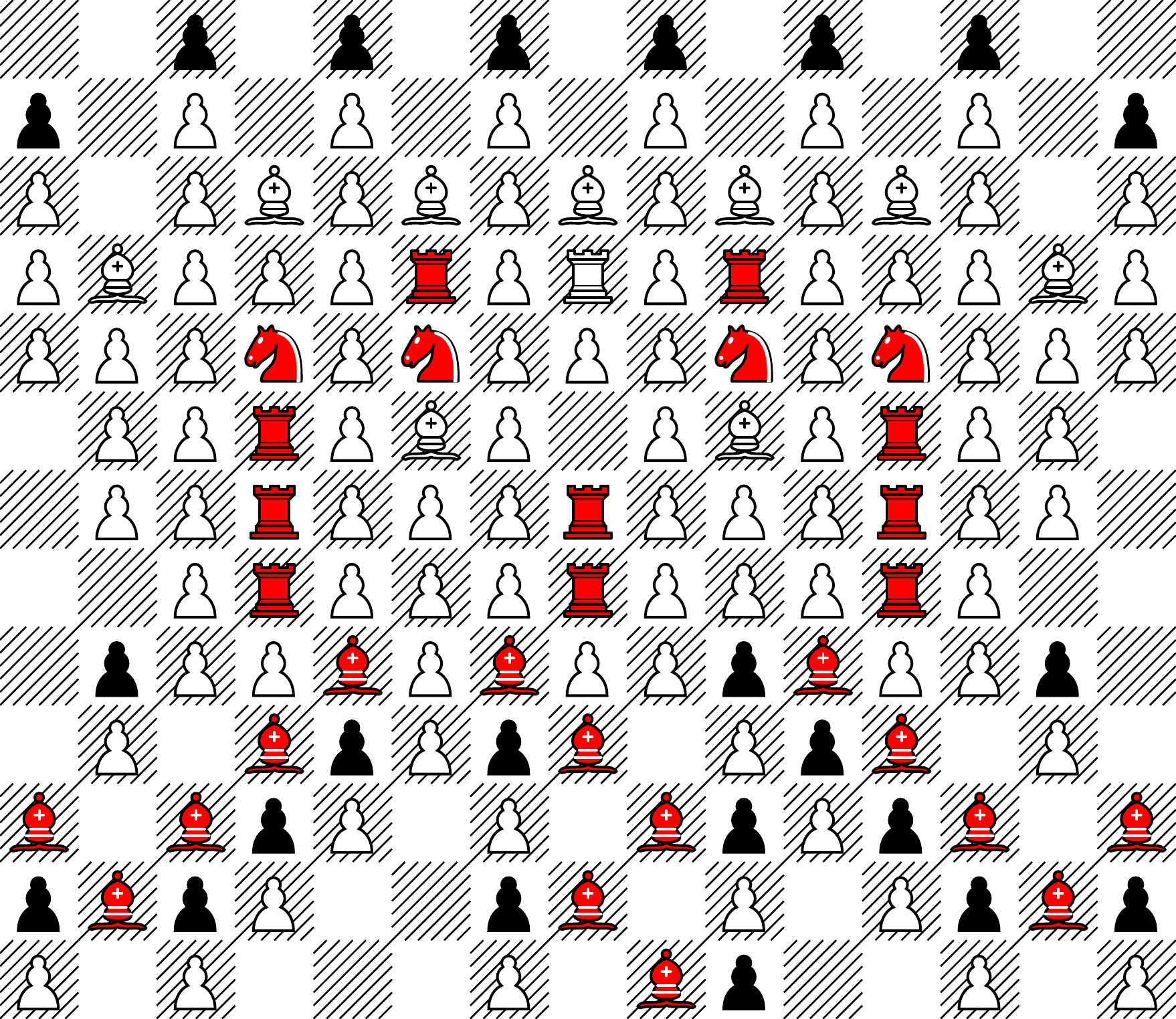}

    \medskip

    \includegraphics[scale=0.75]{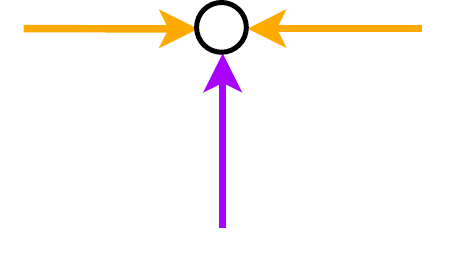}
    \caption{An empty vertex with both orange edges pointing in.}
  \end{subfigure}
  \caption{The vertex gadget. The two edges coming out of the sides are orange edges, and the middle edge coming out of the bottom is purple. Which of two red knights which threaten the center empty space are present determines whether the orange edges are pointing in or out.}
  \label{checkmate vertex}
\end{figure}

Now we move on to the \defn{vertex gadget}. There are two cases for a Subway Shuffle vertex: a vertex which can have two edges of one color both pointing into the vertex when it is empty, and a vertex which has one edge pointing in and one pointing out of the same color when it is empty. Note that a vertex which has all of the edges of a color pointing out does not make sense because a token of that color could never reach the vertex, so those edges are provably unusable. We implement both of these with the same vertex gadget, shown in Figure~\ref{checkmate vertex}. This gadget has three edges coming out of it from the left, right, and bottom. The left and right edges are orange, and the bottom edge is purple. Which type of vertex the gadget represents depends on which red knights are present in the middle. The empty square in the middle is the vertex square. When it contains a knight, it represents a vertex occupied with an orange token, and when it contains a rook, it represents a vertex occupied with a purple token. To use the gadget, white moves all of the red pieces one step away from the vertex along one of the edges until the vertex square becomes empty. This represents a token leaving the vertex along that edge. Then white moves one of the red pieces that can move into the vertex square and continues moving all of the pieces along that path of red pieces; this represents moving a token into this vertex along that edge.

We use the argument about color changing from Lemma~\ref{color changing lemma} to allow all three edges leaving the vertex to be bishop lines. The transition from the rooks in the gadget to the bishop edges is essentially this color-changing. All of the bishops are on dark squares, and our edges can be routed to connect arbitrary squares of same color, so we will not need to worry about parity issues with connecting different vertices.

\begin{figure}
  \centering
  \includegraphics[width=.8\linewidth]{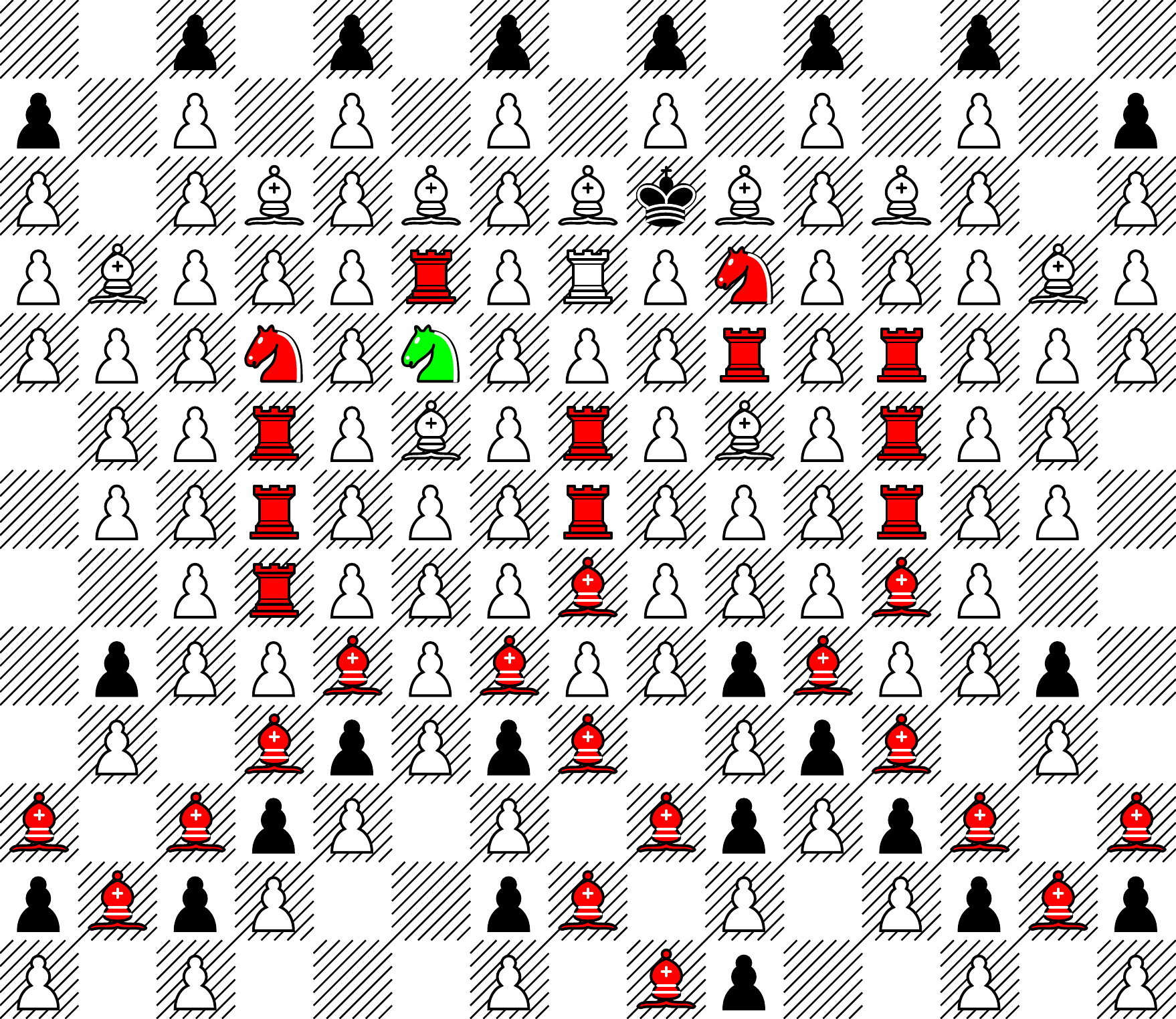}
  \caption{The win gadget. The black king is completely stuck; if it gets checked it will be immediately checkmated. The white knight highlighted in green is the only piece ever capable of accomplishing this. In order to do so, it must first move to the vertex, and then from there move to the right edge.}
  \label{checkmate win}
\end{figure}

Now we have to implement the Subway Shuffle target edge. This means making a \defn{win gadget} which checks whether a particular edge is used. In order for an edge to be used, a piece must leave the vertex at the tail of the edge to move along that edge. Our win gadget is a modified version of the vertex gadget which allows white to checkmate if they can get a knight (representing an orange token) to leave the vertex along a specified edge. Our win gadget is depicted in Figure~\ref{checkmate win}. Note that it is identical to the vertex gadget except for the replacement of one of white's pawns with a black king.

\begin{figure}
  \centering
  \includegraphics[width=.2\linewidth]{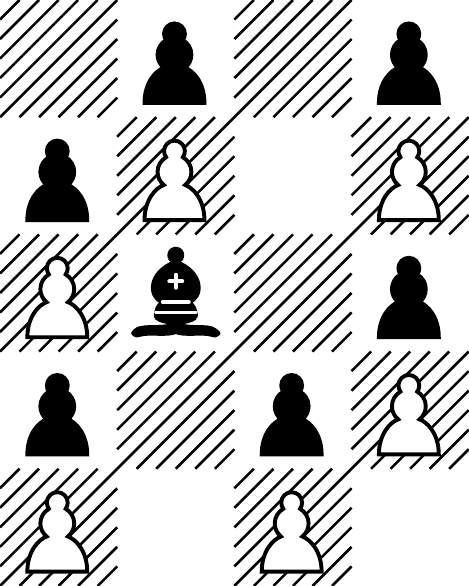}
  \caption{Do-nothing gadget, which allows black to pass forever.}
  \label{checkmate nothing}
\end{figure}

Lastly, we have a \defn{do-nothing gadget}. The entire puzzle is solved by white; all of the pieces that can move in the gadgets are white's and the helpmate in question is white trying to checkmate black. A legal Chess game, however, must have alternating moves by each side. Thus, we need to give black something to do. The gadget in Figure~\ref{checkmate nothing} accomplishes this, by giving black a trapped bishop that they can (and must) move back and forth in between white's moves.

It is worth noting that the positions that result from this reduction are reachable from a starting position, provided we make the board size a polynomial in the size of the Subway Shuffle instance large enough to have enough pieces and pawns to make the gadgets. Any extra pieces can capture each other prior to beginning to construct the position. We can also lock the white king away in a cage similar to the do-nothing gadget.

\subsection{Correctness}

Now we show that the only way white can ever checkmate black is by solving the Subway Shuffle problem and using the gadgets as they are intended to be used.

For the edge gadget, it is easy to check that no piece can move except the red bishops can move one space along the edge when it is in use.

For the vertex gadget, there is one other move white can try, but it does not do anything productive. White can try moving one of the pawns below the rook columns up one space when the rook above it moves up. This results in an immediately stuck position, so it is never useful for white to do. There are no other moves white can legally make outside of the reduction.

\section{Reachability Retrograde Chess Problems are PSPACE-Complete}
\label{sec:reachability}


We reduce from the same problem, max-degree-3 two-color oriented Subway Shuffle, as in the previous section. As in the previous section, the basic structure will have white solving an instance of Subway Shuffle while black effectively passes their turn. But this time, rather than making moves, white will ``undo'' moves, which allows for pawns to move backwards or pieces to be uncaptured, among other things. If white succeeds, the win gadget will allow a piece to escape from the walls of the reduction. Once this hole appears, it will let more pieces forming the walls of the gadgets to start leaving, eventually unravelling all of the gadgets. At this point once the pieces are spread out, it is easy for the players to find a sequence of moves that could get there from the starting position.

Before we describe the gadgets, we will first make some observations about how moves work in retrograde puzzles. Instead of thinking about moves that can be made from a position, we will think about moves that could have just been done; we will call these \emph{unmoves}. All Chess pieces other than pawns unmove the same way that they move. Pawns are different, and all captures are different as well. A piece is never captured in an unmove; to undo a capture, a piece will unmove and the captured piece appears in its place. This means that, unlike in the checkmate reduction before, we will not have to worry about pieces being capturable, so the color of non-pawn pieces in the reduction is irrelevant.

Another important distinction is that pawns do not need a piece to be able to uncapture, so any pawn can always move diagonally backward to uncapture a piece unless the space it would unmove to is occupied. Due to Corollary~\ref{empty file lemma}, this will make walling our gadgets much harder than before since every Chess piece can unmove to some square to its left and some square to its right. This means that we cannot have any isolated gadgets in the middle of the board; in order for any block of pieces to be stuck, the block must extend to both the left and right edges of the board. This results in needing an additional gadget, a terminator that we attach to the ends of the construction which anchors everything to the edge of the board.

\begin{lemma}
  \label{empty adjacent lemma}
  If the five nearest spaces in either file (column) immediately adjacent to a piece are empty, and the piece is not in the first two or last two ranks, then that piece can unmove into that file, leaving an empty space where it came from.
  \end{lemma}
  \begin{proof}
    We simply look at each piece and note that every Chess piece can unmove into a square in the immediately adjacent file. For every non-pawn piece, it simply unmoves there and leaves an empty space immediately. For a pawn, it must uncapture to do this, which it can do because it's not in the first two or last two ranks. It can uncapture a non-pawn piece, and that piece can then unmove into the empty file immediately, leaving a hole.
  \end{proof}

  \begin{corollary}
  \label{empty file lemma}
    If a region of the board which does not include the first two or last two ranks has at least one piece and has an empty file adjacent to it, then a piece can unmove (possibly with multiple unmoves) to escape the region.
  \end{corollary}
  \begin{proof}
    Without loss of generality, let the empty file be on the left. Consider the leftmost piece in the region. Then the conditions of Lemma~\ref{empty adjacent lemma} are satisfied, so the piece can unmove into that file. From here the piece can continue unmoving until it leaves the region.
  \end{proof}

\begin{figure}
  \centering
  \includegraphics[width=.5\linewidth]{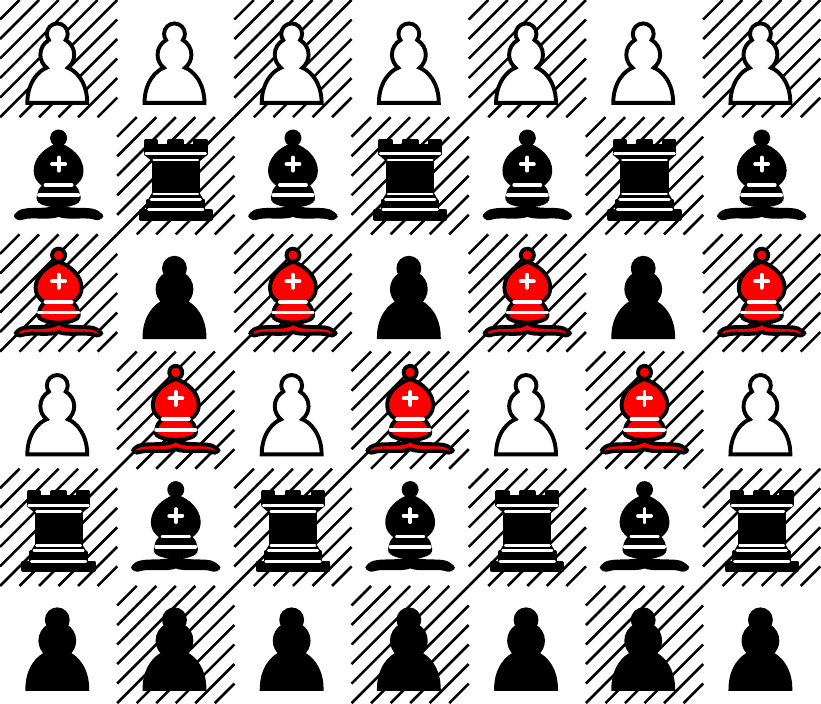}
  \caption{The edge gadget.}
  \label{retro edge}
  \end{figure}

\subsection{Gadgets}

Now we describe the gadgets in our reduction.

First is the \defn{edge gadget} shown in Figure~\ref{retro edge}. Like in the previous section, this gadget uses a line of bishops each of which move one space to represent the movement of a token. To keep the bishops locked in, we use a repeating pattern of pawns, rooks, and bishops across the top and bottom of the edge. Because pawns care about the orientation of the board, we cannot actually make vertical edges. Instead, we use the \defn{turn} and \defn{shift gadgets} shown in Figure~\ref{retro turn shift}; if you wiggle an edge back and forth with turns and shifts you can make it travel vertically up the board.

\begin{figure}
  \centering
  \begin{subfigure}{.32\linewidth}
    \centering
    \includegraphics[width=\linewidth]{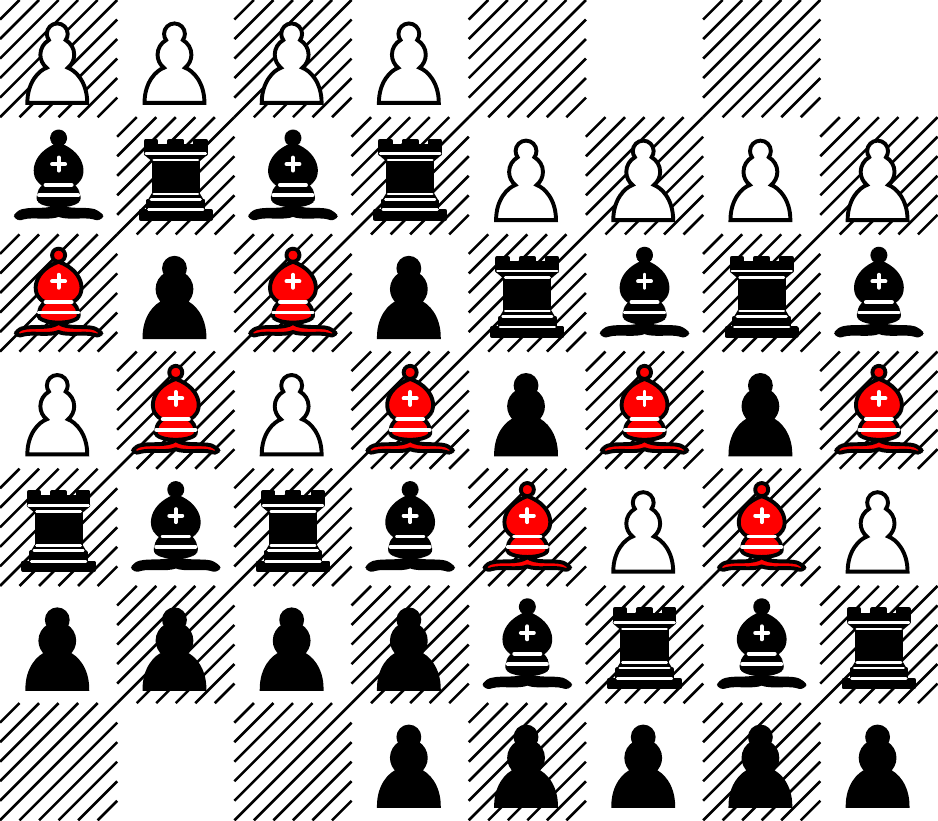}
    \caption{A shifted edge.}
  \end{subfigure}%
  \hfill
  \begin{subfigure}{.62\linewidth}
    \centering
    \includegraphics[width=\linewidth]{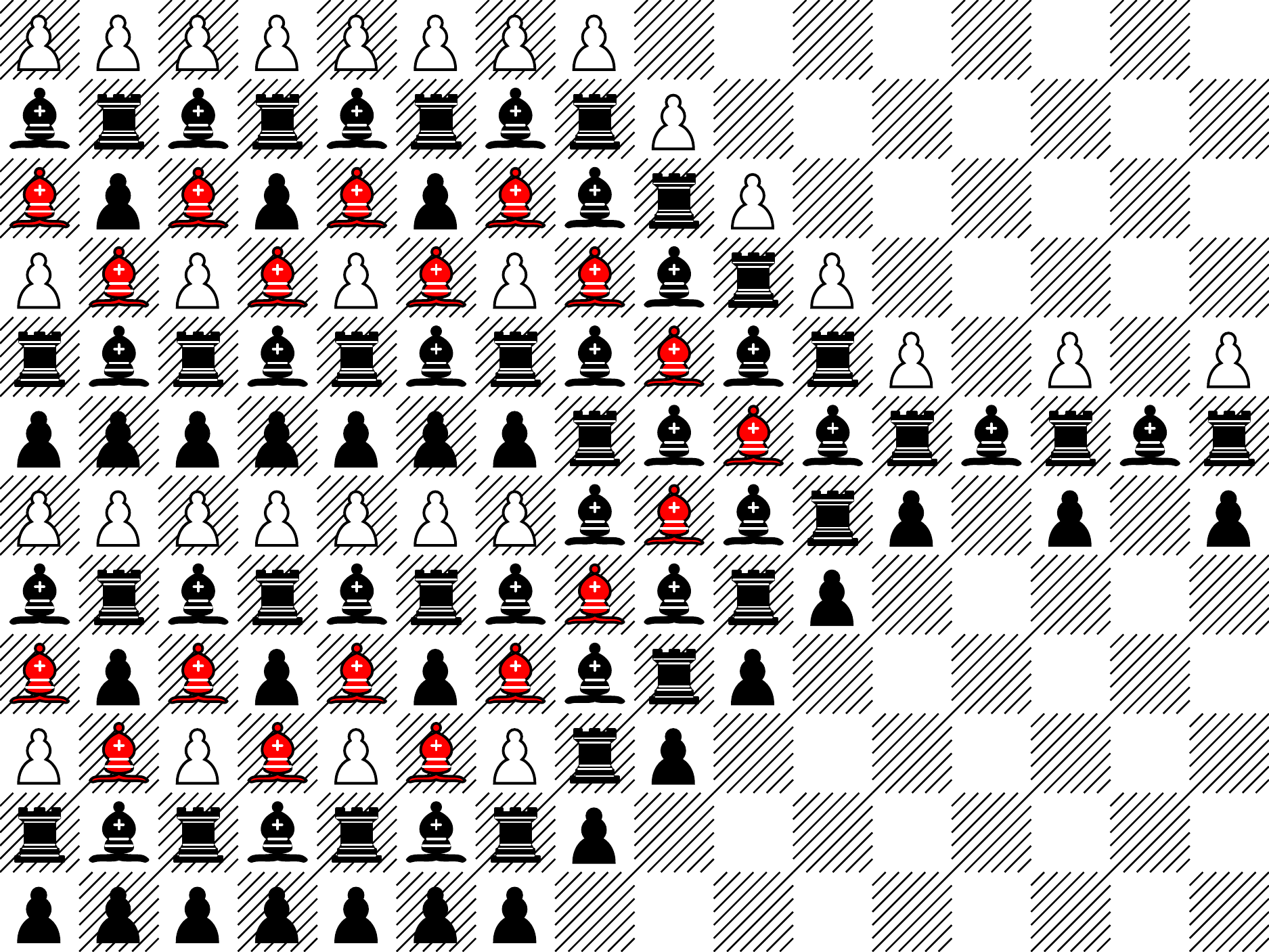}
    \caption{A U-turn gadget for edges.}
  \end{subfigure}
  \caption{Edge routing gadgets: shift and U-turn.}
  \label{retro turn shift}
\end{figure}

\begin{figure}
  \centering
  \includegraphics[width=.8\linewidth]{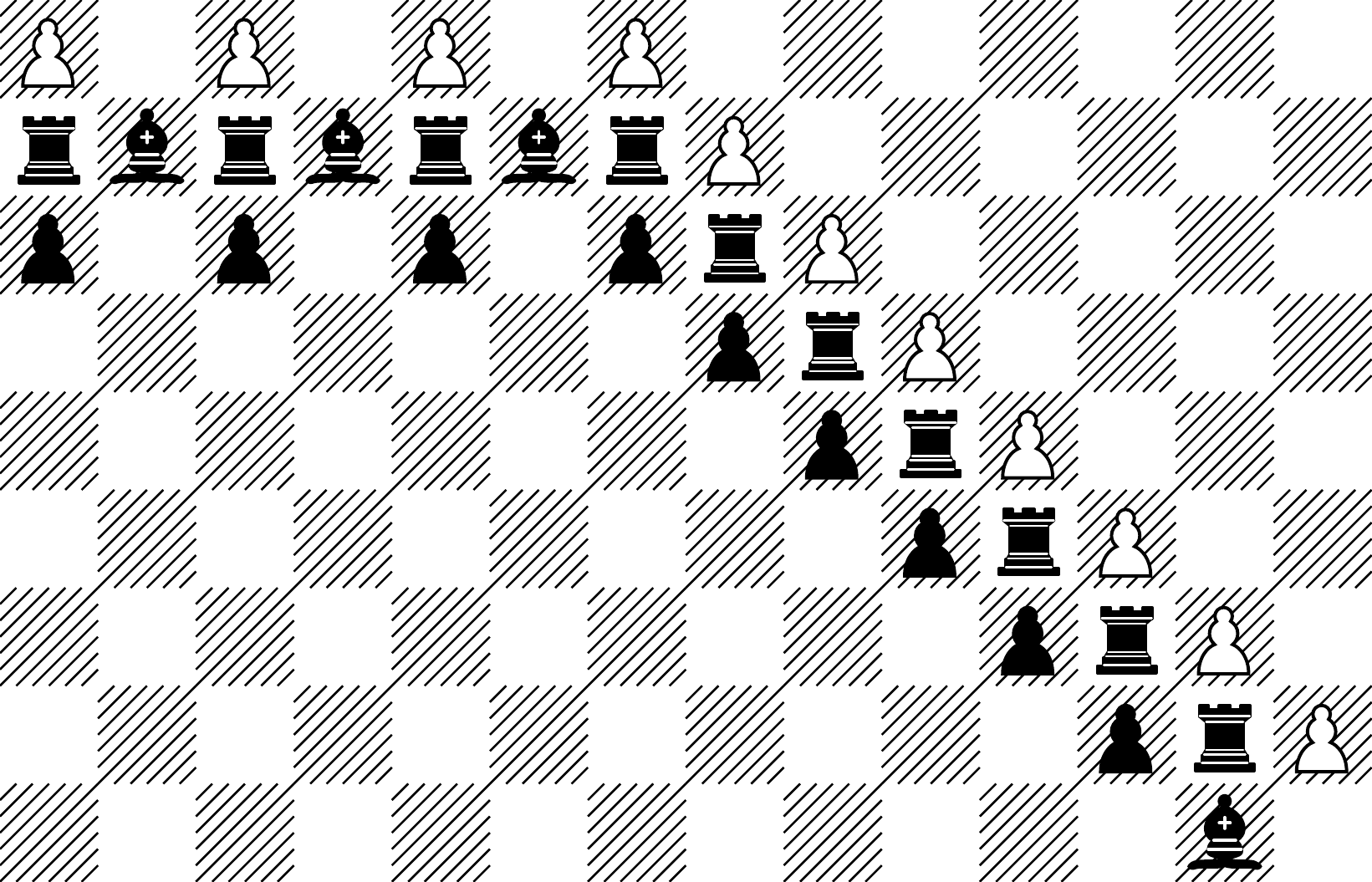}
  \caption{Terminator gadget that connects the loose ends of gadgets to the bottom rank. Unlike other figures, here we care that the first rank of this figure is the actual first rank of the Chess board. In particular, this means that the white pawn on the second rank cannot unmove, and that allows us to prove that everything is stuck.}
  \label{retro terminator}
\end{figure}

At the edge of the turn gadget, we have the \defn{terminator gadget}, shown in Figure~\ref{retro terminator}. As previously stated, because every piece is capable of unmoving to both adjacent files, we need a terminator gadget which connects these loose ends to the edge of the board. We have all of our terminator gadgets terminate either on the first rank of the board or on any other gadget. The terminator gadget in Figure~\ref{retro terminator} terminates on the first rank. To have one terminate on another gadget, it simply runs (diagonally) into the wall of pawns on either side of any of our gadgets, including another terminator. We will have only a single terminator gadget on each side of the construction terminate on the first rank, and all others will terminate on another gadget.

\begin{figure}
  \centering
  \includegraphics[width=.8\linewidth]{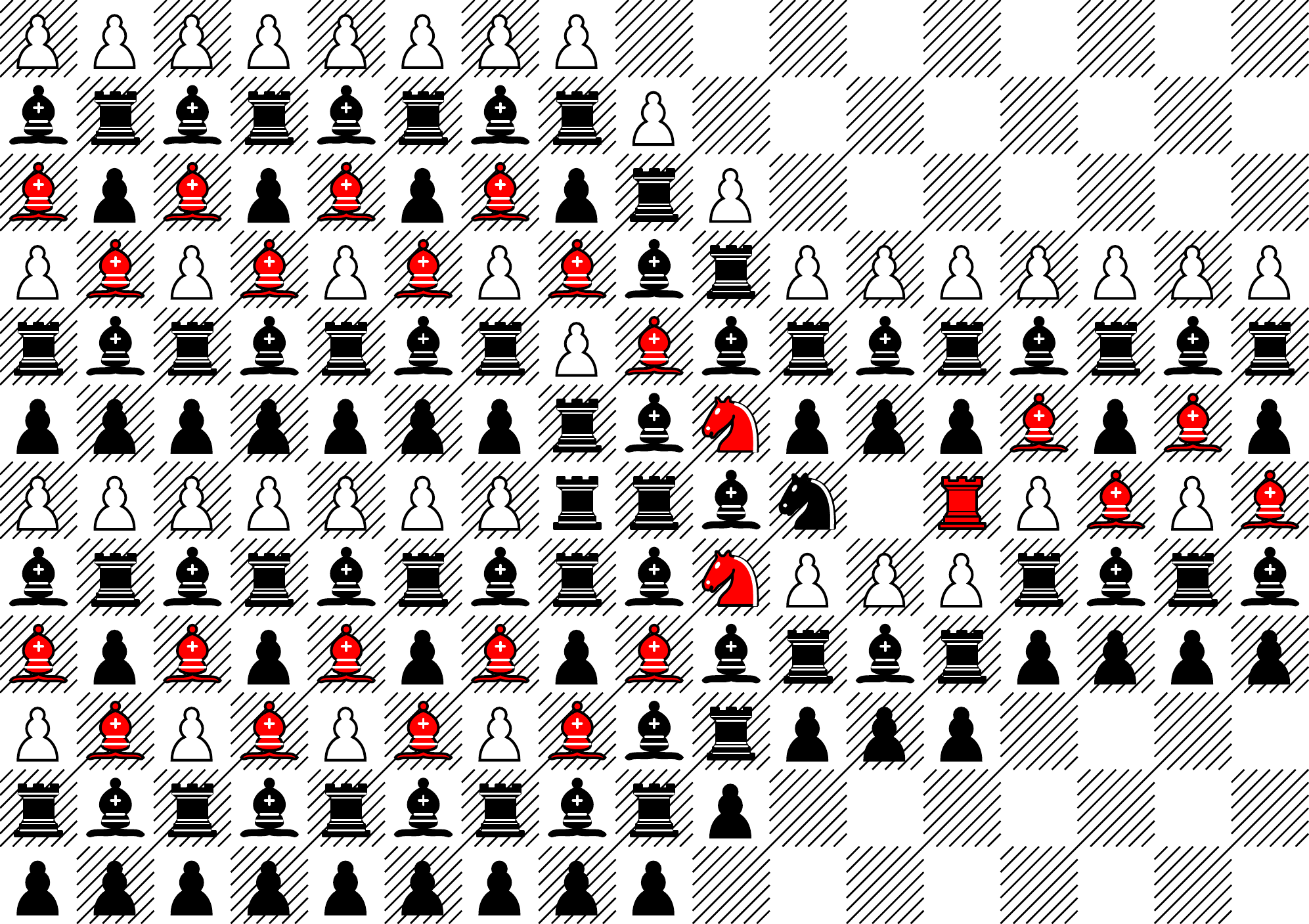}
  \caption{The vertex gadget. The empty square in the middle is the vertex; whether it is occupied by a knight or a rook determines the color of the token at this vertex.}
  \label{retro vertex}
\end{figure}

Now we describe the \defn{vertex gadget}, shown in Figure~\ref{retro vertex}. This vertex has a similar structure to the vertex gadget from the previous section, with potentially two knights and a rook representing the three tokens that can move in from connecting edges into the vertex. The two edges connected by a knight to the vertex are the orange edges; the rook is the purple edge. When both knights are present, we get a vertex which has both orange edges pointing into the vertex. If only one knight is present and the other is replaced by a bishop, only the edge with the knight points into the vertex and the other orange edge points out of the vertex.

\begin{figure}
  \centering
  \includegraphics[width=.8\linewidth]{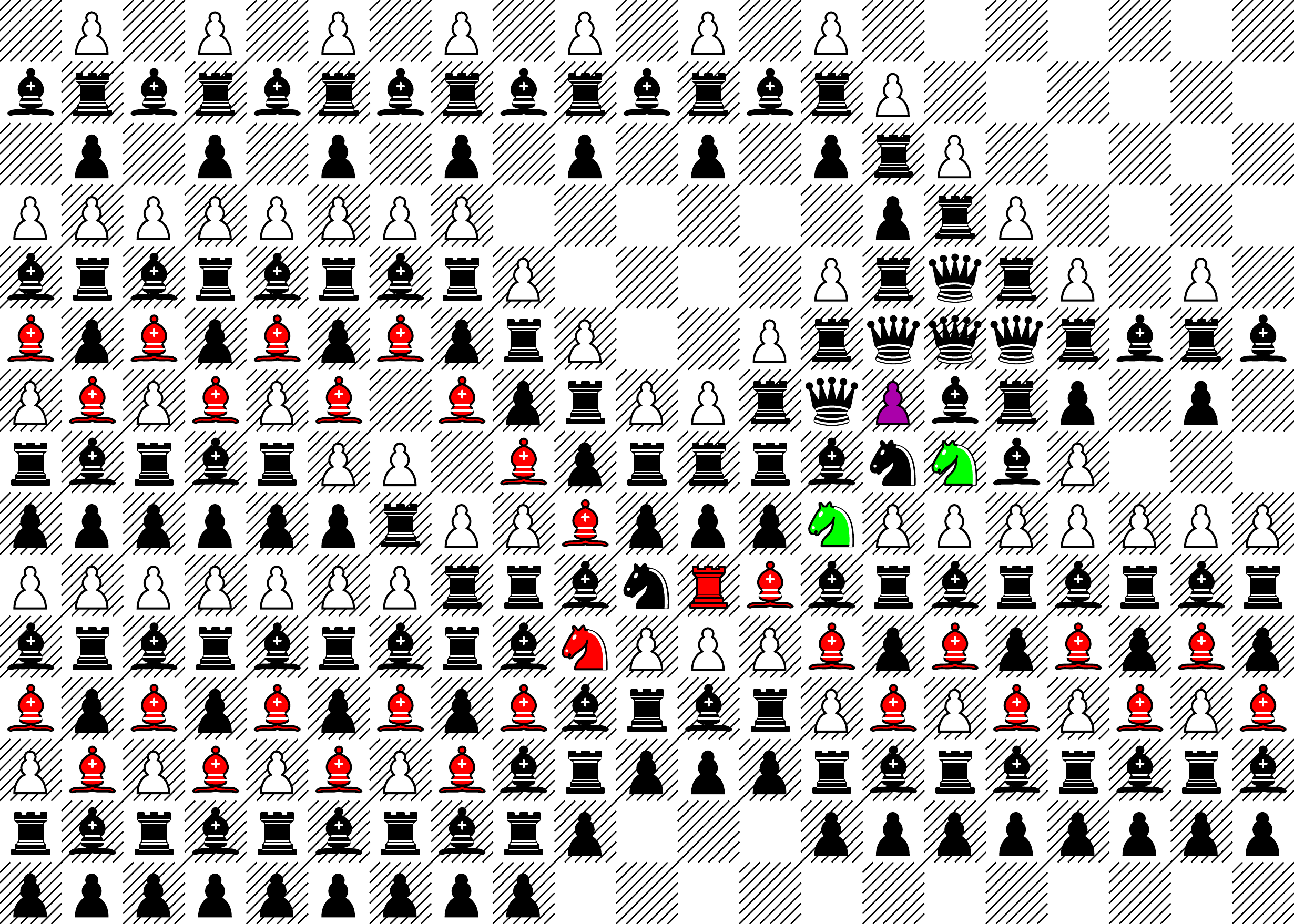}
  \caption{The win gadget. The purple pawn is a white pawn. If a knight leaves the vertex by the top edge, it allows the two white knights highlighted in green to follow it. Then the purple pawn can uncapture a knight where a knight was, which can move away, starting the unravelling process.}
  \label{retro win}
\end{figure}

\begin{figure}
  \centering
  \includegraphics[width=.8\linewidth]{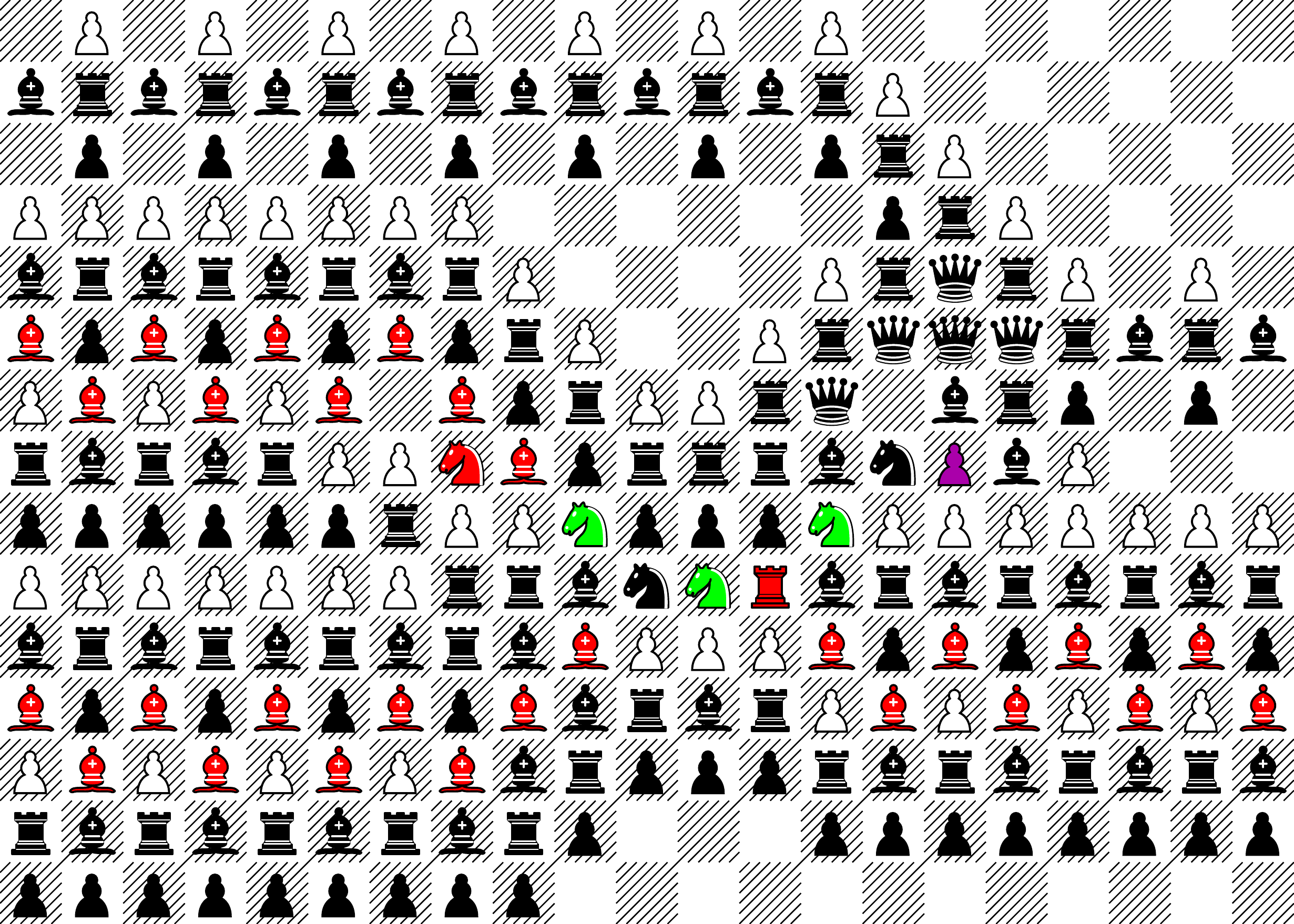}
  \caption{The win gadget after the first few unmoves to begin the unravelling are made. From here, the queens can leave allowing the terminator gadget in the top right to start unravelling.}
  \label{retro start winning}
\end{figure}

\subsection{Win Gadget and Self-Destruction}

Finally we have the \defn{win gadget} shown in Figure~\ref{retro win}. Here we have modified the vertex gadget to add an extra hole a knight's move away from the top orange edge's connection to the vertex. When the top edge is used by having a knight enter it from the vertex, it can hop into this hole. This creates a second hole, which will be key to allowing the construction to unravel. Protruding from the top left and top right of the gadget are two terminator gadgets. The unravelling of these will be crucial to winning.

Normally, the green knight just to the right of the vertex is capable of unmoving to the vertex when it is empty. After this the second green knight can follow it unmoving into the square it just left. This then allows the purple (white) pawn to uncapture the square the knight was on. However, regardless of which piece the pawn uncaptures, the new piece is completely stuck and incapable of moving.

However, if there were a second hole for the green knights to jump into, then once the green knights unmove again, the purple pawn uncaptures a knight, which can then unmove to where one of the knights was. This is shown in Figure~\ref{retro start winning}. At this point, the queens and rooks can shift around to let the queens start escaping. From here, everything begins to unravel. This is where the two terminator lines come in. Once a few queens leave, both of these lines can start to unravel.

We choose a layout of the Subway Shuffle instance such that the win vertex is the furthest north vertex, and these two terminator lines are on the outside face of the graph. We extend these lines very far away from the rest of the construction, which is possible because the choice of layout implies no other part of the construction is in as high a rank as win vertex. We have any other terminator lines from U-turn gadgets which have not already terminated on another gadget terminate on these two terminator lines. Only these two terminators will eventually reach the first rank of the board, as shown in Figure~\ref{retro terminator}. 

Once these two lines have unravelled, it is now the case that our construction is in a region in the middle of the board with no pieces on either side of it. This means we can repeatedly apply Corollary~\ref{empty file lemma}, until every piece has left the construction. It is fairly easy once all of the pieces are free in the middle of the board to find a sequence of unmoves to send them home.

\subsection{Counting Pieces}

Now we need to do a piece counting argument, to show that the position is even plausible. One property of a starting Chess position is that it has only one pawn of each color in each file. Not only this, but pawns also cannot stray too far from their starting file. In particular, a pawn on rank $n$ must come from a file at most $n$ files away from its current position. Unfortunately, our construction can have many pawns in each file, and furthermore is constrained in how far it can be away from the bottom edge of the board due to the terminator gadgets. 

However, the terminator gadget has the property that along the horizontal part, it has a white pawn (and similarly black pawn) density of only one pawn every two files. If we stretch the horizontal part far enough, and put the construction on a sufficiently far forward rank, we can use this to get the pawn density below one pawn per file. As long as this area with low pawn density is at a higher rank on the board than the number of files it is wide, with enough uncaptures the pawns can sort themselves into one pawn per file. The number of uncaptures required is at most quadratic in the number of pawns.

We also need to check the number of non-pawn pieces. To make sure that the board has the right amount of pieces, we simply have the board be much larger than our construction. Our pawns will need to make a large number of uncaptures during the unravelling, and to handle this we will have the board be much larger than the number of pieces in our construction. It is always possible to keep uncapturing additional pieces and pawns so we do not need to worry about having too large of a board.

Finally, every legal Chess position needs to have one king of each color. Since we don't use kings anywhere in the construction, and their ability to roam free doesn't allow the players to unravel the position without solving the Subway Shuffle instance, we simply put the two kings in their home positions. This also ensures that both players always have legal unmoves allowing white and black to alternate making unmoves as is required in Chess.

\section*{Acknowledgments}

This work was initiated during open problem solving in the MIT class on
Algorithmic Lower Bounds: Fun with Hardness Proofs (6.892) in Spring 2019.
We thank the other participants of that class --- in particular, John Urschel
--- for related discussions and providing an inspiring atmosphere.

\bibliographystyle{alpha}
\bibliography{paper}

\end{document}